\documentclass[12 pt]{article}
\usepackage{eurosym}
\usepackage{amssymb}
\usepackage{amsmath}
\usepackage{amsfonts}
\usepackage{graphicx}
\usepackage{geometry}
\usepackage{array}
\usepackage{xcolor}
\usepackage{tikz}
\usepackage{colortbl}
\usepackage{graphicx}
\usepackage{booktabs}
\usepackage{url}

\usetikzlibrary{arrows.meta,automata,positioning}
\newtheorem{theorem}{Theorem}

\newtheorem{example}{Example}

\newtheorem{proposition}{Proposition}
\newtheorem{remark}{Remark}

\newenvironment{proof}[1][Proof]{\noindent\textbf{#1.} }{\ \rule{0.5em}{0.5em}}

\usepackage{setspace}

\begin{document}

\title{Monotonicity in streaming problems\thanks{%
Financial support from grants PID2023-146364NB-I00 and PID2025-168226NB-I00,
funded by MICIU/AEI/10.13039/501100011033/ and by FEDER, and UE is
gratefully acknowledged.}}
\author{\textbf{Gustavo Berganti\~{n}os}\thanks{%
ECOBAS, Universidade de Vigo, ECOSOT, 36310 Vigo, Espa\~{n}a;
gbergant@uvigo.es} \\
\textbf{Juan D. Moreno-Ternero}\thanks{%
Department of Economics, Universidad Pablo de Olavide, 41013 Sevilla, Espa%
\~{n}a; jdmoreno@upo.es}}
\maketitle

\begin{abstract}
We explore the implications of the principle of monotonicity in the problem
of sharing the revenues raised from subscriptions to music streaming
platforms among content providers. We formalize various forms of this
principle as several axioms for popularity indices, which generate sharing
rules in this setting. We show that, combined with some other basic axioms
(such as additivity, symmetry, or pairwise homogeneity), they provide
axiomatic characterizations of focal indices for this problem, 
as well as more general families comprising them. 
\end{abstract}

\bigskip

\medskip {} 


\bigskip

\bigskip

\newpage

\bigskip

\bigskip

\section{Introduction}

The music industry underwent major changes since the advent of digital music
and peer-to-peer technologies at the end of the 20th century. Streaming
emerged some years later as a breakthrough, which generated new business
opportunities, including the creation of platforms, key players in the
industry nowadays (e.g., Aguiar, 2017; Aguiar et al., 2024). The most famous
case is Spotify, with 777 million monthly active users worldwide, among
which around 300 million people pay for a premium subscription.\footnote{%
https://newsroom.spotify.com/company-info/} A common practice for streaming
platforms is to distribute around $70\%$ of the revenue received from
subscriptions among artists uploading content therein (e.g., Meyn et al.,
2023). Such an allocation problem, akin to others emerging from related
bundling practices (e.g., Ginsburgh and Zang, 2003; Berganti\~{n}os and
Moreno-Ternero, 2015) 
is the object of study in this paper.

We consider the stylized model introduced by Berganti\~{n}os and
Moreno-Ternero (2025a), whose ingredients are three: a group of artists, a
group of users, and a matrix indicating the streaming times each user
streamed each artist. Based on this input, a popularity \textbf{index}
(function), which measures the importance of each artist, is constructed.
The shares received by each artist from the revenues generated in each
problem are based on such an index. This model has been studied in Berganti%
\~{n}os and Moreno-Ternero (2025a, 2025b, 2026a, 2026b), Gon\c{c}%
alves-Dosantos et al. (2025a, 2025b), Dietzenbacher et al. (2026),
Dietzenbacher and N\'{u}\~{n}ez-Lugilde (2026) and Feng and Ma (2026). 
The novelty here is to formalize various forms of the principle of \textit{%
monotonicity}, as several axioms for popularity indices, and to explore
their implications in this setting.

Monotonicity is a general principle of fair division with a long tradition
in the literature.\footnote{%
Early instances are Megiddo (1974), Kalai and Smorodinsky (1975), Kalai
(1977), Thomson and Myerson (1980), Young (1985, 1987, 1988), Roemer (1986),
Chun and Thomson (1988), Moulin and Thomson (1988), or Thomson (1999). More
recent instances are Berganti\~{n}os and Moreno-Ternero (2022) or Abe et al.
(2025).} It states that when the underlying data of a problem change in a
specific way, the solution should change accordingly. A typical formulation
is as follows. Let $P$ and $\widehat{P}$ be two problems such that the
situation of agent $i$ at $P$ is \textquotedblleft better" than at $\widehat{%
P}$. Then, the allocation for agent $i$ at $P$ should not be worse than at $%
\widehat{P}$. Depending on the specific meaning of \textquotedblleft
better", various monotonicity axioms could be defined. 
We shall consider several meanings in our setting, giving rise to the
following specific axioms:



\emph{Overall monotonicity}: when all entries in the streaming matrix
increase.


\emph{Artist monotonicity}: when the streams of an artist increase, user by
user, and the streams of the other artists remain the same.


\emph{Total artist monotonicity}: when the total streams of an artist
increase and the total streams of the other artists remain the same.


\emph{Reciprocal artist monotonicity}: when the streams of an artist remain
the same, whereas the streams of the other artists decrease.


\emph{User monotonicity}: when the streams of a user increase, artist by
artist, and the other users remain the same.

\bigskip

The most frequent sharing methods to solve streaming problems (more
precisely, the allocation problem underlying the revenues raised by
streaming) are the so-called \textbf{pro-rata} method, which rewards artists
proportionally to the total number of streams and the so-called \textbf{%
user-centric} method, which rewards artists so that the revenue generated by
each user is divided among the artists streamed by the user, proportionally
to the total number of streams (e.g., Alaei et al., 2022). One of the side
messages of our analysis here is that the former performs better with
respect to monotonicity properties than the latter. This contrasts with
earlier axiomatic studies for streaming problems, which favor the
user-centric index with respect to the pro-rata index (e.g., Berganti\~{n}os
and Moreno-Ternero, 2025a) but it aligns with other (non-axiomatic) studies
favoring the latter (e.g., Alaei et al., 2022; Yu, 2026).

Another (admittedly, somewhat naive) method that performs extremely well
with respect to monotonicity properties is the \textbf{egalitarian} method
(e.g., Gon\c{c}alves-Dosantos et al., 2025a,b), which simply rewards all
artists equally, independently of streams. This should not be a surprise, as
some monotonicity properties have strong egalitarian implications in related
settings (e.g., Flores-Szwagrzak et al., 2026). Finally, we also explore the
so-called \textbf{Shapley} method (e.g., Berganti\~{n}os and Moreno-Ternero,
2025b, Gon\c{c}alves-Dosantos et al., 2025a), which performs similarly to
the user-centric index with respect to monotonicity axioms.

The main outcome of our paper is to provide several characterization results
thanks to the monotonicity axioms. First, we show the \textit{pro-rata}
index is characterized by \textit{additivity} (if we can present a problem
as the sum of smaller problems, then the index in the original problem
should coincide with the sum of the indices in the smaller problems), 
\textit{pairwise homogeneity} (if each user has streamed an artist a certain
times more than another artist, then the index should preserve that ratio)
and \textit{total artist monotonicity}, or \textit{reciprocal artist
monotonicity} and \textit{anonymity on users} (the name of users does not
matter). Second, the \textit{Shapley} index is characterized by \textit{%
additivity}, \textit{symmetry} (if the streams of two artists coincide for
each user, then the index for both artists should coincide), \textit{budget
balance} (the index fully divides the overall budget), \textit{null artist}
(if an artist has zero streams, then the importance of this artist should be 
$0$), \textit{non-unilateral manipulability} (an artist cannot increase the
index unilaterally upon increasing its number of streams from fans) and 
\textit{artist monotonicity}. Finally, the \textit{egalitarian} index is
characterized by \textit{symmetry}, \textit{budget balance} and either 
\textit{overall monotonicity}, \textit{user monotonicity} or \textit{%
reciprocal artist monotonicity}.

We extend our analysis, to uncover further the structure of the space of
indices, by exploring the implications of each of the monotonicity axioms
within the family of \textit{decomposable indices} (e.g., Berganti\~{n}os
and Moreno-Ternero, 2026b), the most general family of indices in the
literature, simply characterized by the two basic axioms of \textit{%
additivity} and \textit{null artist}.

The rest of the paper is organized as follows. In Section 2, we introduce
the model and the basic definitions (including those of axioms and indices),
as well as some running examples. In Section 3, we provide the
characterization results for specific indices described above. In Section 4,
we provide the characterization results for the sub-families of decomposable
indices emerging after imposing each of the monotonicity axioms in our
paper. We conclude in Section 6.

\section{Streaming problems}

We consider the model introduced by Berganti\~{n}os and Moreno-Ternero
(2025a). Let $\mathbb{N}$ represent the set of all potential artists and $%
\mathbb{M}$ the set of all potential users (of music streaming platforms). 
Each platform involves a (finite) set of artists $N\subset \mathbb{N}$ and a
(finite) set of users $M\subset \mathbb{M}$. We typically assume that $%
\mathbb{N}$ and $\mathbb{M}$ coincide with the set of positive integer
numbers (namely, $\mathbb{Z}_+=\left\{ 1,2,...\right\} )$, $N=\left\{
1,...,n\right\} $ and $M=\left\{ 1,...,m\right\} $. For each $(i,j)\in
N\times M$, let $t_{ij}$ denote the times user $j$ streamed contents
uploaded by artist $i$ in the platform, during a certain period of time
(briefly, \textit{streams}). In most platforms, playing a streaming unit is
equivalent to playing a song (for at least 30 seconds). Let $t=\left(
t_{ij}\right) _{i\in N,j\in M}$ denote the matrix of all streams. We assume
that for each $j\in M,$ $\sum\limits_{i\in N}t_{ij}>0$ (namely, each user
has streamed some content).

A \textbf{streaming problem} is a triple $P=\left( N,M,t\right) $. We
normalize the amount paid by each user to $1$. Consequently, the amount to
be divided among artists in a problem is $m$, the number of users. The set
of problems so defined is denoted by $\mathcal{P}$.

Let $\left( N,M,t\right)\in\mathcal{P}$. For each $j\in M,$ we denote by $%
t^{-j}$ the matrix obtained from $t$ by removing the column corresponding to
user $j$. Similarly, for each $i\in N,$ we denote by $t_{-i}$ the matrix
obtained from $t$ by removing the row corresponding to artist $i$. For each
artist $i\in N$, we denote the total number of streams of artist $i$ as $%
T_{i}\left( N,M,t\right) =\sum_{j\in M}t_{ij}$. For each user $j\in M$, we
denote the total number of streams of user $j$ as $T^{j}\left( N,M,t\right)
=\sum_{i\in N}t_{ij}$. 
We define the list of artists of a user as those from which the user has
streamed content. Formally, for each $j\in M$, $L^{j}\left( N,M,t\right)
=\left\{ i\in N:t_{ij}>0\right\}$. The profile of user $j$ is defined as the
streaming vector associated with the user. Namely, $t_{.j}\left(
N,M,t\right) =\left( t_{ij}\right) _{i\in N}$. When no confusion arises we
write $T_{i}$ instead of $T_{i}\left( N,M,t\right) $, $T^{j}$ instead of $%
T^{j}\left( N,M,t\right)$, 
$L^{j}$ instead of $L^{j}\left( N,M,t\right) ,$ and $t_{.j}$ instead of $%
t_{.j}\left( N,M,t\right)$.

\bigskip

We illustrate our model with a basic example to which we shall turn to
throughout our ensuing analysis. Assume two users join a platform to listen
to their favorite artists. Assume artists $1$ and 2 are streamed by user $1$%
, whereas artist $3$ is streamed by user $2$. This situation can be included
in our theoretical model as follows.

\begin{example}
\label{ex 2,3,a} Let $N=\left\{ 1,2,3\right\} ,$ $M=\left\{ 1,2\right\} ,$
and 
\begin{equation*}
t=\left( 
\begin{array}{cc}
10 & 0 \\ 
20 & 0 \\ 
0 & 70%
\end{array}%
\right) .
\end{equation*}
\end{example}

\medskip

A popularity \textbf{index} is a mapping that associates a number to each
artist in each problem. Formally, $I:\mathcal{P}\rightarrow
\bigcup\limits_{N\subset \mathbb{N}}\mathbb{R}_{+}^{n}$ is such that for
each $\left( N,M,t\right) \in \mathcal{P},$ $I\left( N,M,t\right) \in 
\mathbb{R}_{+}^{n}$. Besides, for each pair $i,j\in N$, we are interpreting
that $I_{i}\left( N,M,t\right) \geq I_{j}\left( N,M,t\right) $ if and only
if $i$ is at least as important as $j$ at problem $\left( N,M,t\right) $. We
assume that $\sum\limits_{i\in N}I_{i}\left( N,M,t\right) >0$.

The \textbf{reward} received by each artist $i\in N$ in each problem is
proportional to the importance of that artist in that problem. Formally, 
\begin{equation*}
R_{i}^{I}\left( N,M,t\right) =\frac{I_{i}\left( N,M,t\right) }{%
\sum\limits_{i^{\prime }\in N}I_{i^{\prime }}\left( N,M,t\right) }m.
\end{equation*}

As for each $\lambda >0$ and each index $I$, $R^{\lambda I}=R^{I}$, we
identify an index with all its positive linear transformations, unless
explicitly stated otherwise.







Berganti\~{n}os and Moreno-Ternero (2026b) introduce the most general family
of indices studied in the literature. Formally, we say that $I$ is a \textbf{%
decomposable index} if the following two conditions hold:

\begin{enumerate}
\item There exists a function $d:\mathbb{N}\times \mathbb{M}\times
\bigcup\limits_{N\subset \mathbb{N}}\left( \mathbb{Z}_{+}\cup \left\{
0\right\} \right) ^{n}\rightarrow \mathbb{R}_{+}$ such that for each $i\in
N\subset \mathbb{N}$, each $j\in \mathbb{M}$, and each $x\in \left(\mathbb{Z}%
_{+}\cup \left\{ 0\right\} \right) ^{n}$, $d\left( i,j,x\right) =0$ whenever 
$x_{i}=0.$

\item For each $\left( N,M,t\right) \in \mathcal{P}$ and each $i\in N$, 
\begin{equation*}
I_{i}\left( N,M,t\right) =\sum_{j\in M}d\left( i,j,t._{j}\right) .
\end{equation*}
\end{enumerate}

For each function $d$ satisfying conditions $1$ and $2$, we denote by $I^{d}$
the decomposable index associated with $d$.

Three decomposable indices stand out. 

First, the so-called \textbf{pro-rata} index, which simply measures
importance by the total number of streams. Formally, for each problem $%
\left( N,M,t\right) \in \mathcal{P}$ and each artist $i\in N,$ 
\begin{equation*}
P_{i}\left( N,M,t\right) =T_{i}=\sum_{j\in M}t_{ij}.
\end{equation*}

Thus, the amount received by each artist $i\in N$ under this index is 
\begin{equation*}
R_{i}^{P}\left( N,M,t\right) =\frac{T_{i}}{\sum\limits_{k\in N}T_{k}}m.
\end{equation*}

Second, the so-called \textbf{user-centric} index. The amount paid by each
user is divided among the artists streamed by the user, proportionally to
the total number of streams. Then, the index assigns each artist the sum of
the amounts obtained across users. Formally, for each problem $\left(
N,M,t\right) $ and each artist $i\in N,$ 
\begin{equation*}
U_{i}\left( N,M,t\right) =\sum_{j\in M}\frac{t_{ij}}{T^{j}}.
\end{equation*}%
It is easy to see that $R^{U}\left( N,M,t\right) =U\left( N,M,t\right) .$

\bigskip

Third, the so-called \textbf{Shapley} index. It divides the amount paid by
each user among the artists streamed by the user equally among these artists
(thus, ignoring the streaming times; as opposed to the user-centric index,
which did it proportionally to the streaming times). Formally, for each
streaming problem $\left( N,M,t\right) \in \mathcal{P}$ and each $i\in N, $ 
\begin{equation*}
Sh_{i}\left( N,M,t\right) =\sum_{j\in M:i\in L^{j}}\frac{1}{\left\vert
L^{j}\right\vert }.
\end{equation*}%
As $\sum\limits_{i\in N}Sh_{i}\left( N,M,t\right) =m$, the index and the
associated reward coincide too. Namely, $Sh\left( N,M,t\right)
=R^{Sh_{i}\left( N,M,t\right) }$.

Berganti\~{n}os and Moreno-Ternero (2026b) show that decomposable indices
are characterized by the following two basic axioms.




\bigskip


\textbf{Additivity}. For each triple $\left( N,M^{1},t^{1}\right) ,\left(
N,M^{2},t^{2}\right) ,\left( N,M,t\right) \in \mathcal{P}$ such that $%
M=M^{1}\cup M^{2}$, $M^{1}\cap M^{2}=\varnothing ,$ $t_{ij}=t_{ij}^{1}$ when 
$j\in M^{1}$ and $t_{ij}=t_{ij}^{2}$ when $j\in M^{2},$ 
\begin{equation*}
I\left( N,M,t\right) =I\left( N,M^{1},t^{1}\right) +I\left(
N,M^{2},t^{2}\right) .
\end{equation*}



\textbf{Null artist}. For each $\left( N,M,t\right) \in \mathcal{P}$, and
each $i\in N$ such that $T_{i}=0$, 
\begin{equation*}
I_{i}\left( N,M,t\right) =0.
\end{equation*}

Another index that will be salient in our analysis is the \textbf{egalitarian%
} index, which simply states that all artists have the same importance. This
index obviously violates null artist and, thus, it lies outside the family
of decomposable indices.

Formally, for each problem $\left( N,M,t\right) \in \mathcal{P}$ and each
artist $i\in N,$ 
\begin{equation*}
E_{i}\left( N,M,t\right) =\frac{m}{n}.
\end{equation*}

To illustrate how these four indices work, we resort to Example \ref{ex
2,3,a}. 
\begin{equation*}
\begin{tabular}{cccccc}
& $P\left( N,M,t\right) $ & $R_{i}^{P}\left( N,M,t\right) $ & $U_{i}\left(
N,M,t\right) $ & $Sh_{i}\left( N,M,t\right) $ & $E_{i}\left( N,M,t\right) $
\\ 
1 & 10 & 0.2 & 0.33 & 0.5 & 0.67 \\ 
2 & 20 & 0.4 & 0.67 & 0.5 & 0.67 \\ 
3 & 70 & 1.4 & 1 & 1 & 0.67%
\end{tabular}%
\end{equation*}

\bigskip

\subsection{Monotonicity axioms}

The aim of this paper is to explore the implications in this setting of the
principle of monotonicity, which has played an important role in axiomatic
work. We now consider several axioms modelling various forms of the
principle. First, the standard axiom saying that the index should be
monotonic on the entries of $t$. Formally,

\textbf{Overall monotonicity}. For each pair $\left( N,M,t\right) ,\left(
N,M,t^{\prime }\right) \in \mathcal{P}$ such that $t\geq t^{\prime }$,%
\footnote{%
Given two matrices $t,t^{\prime }$ we write $t\geq t^{\prime }$ whenever $%
t_{ij}\geq t_{ij}^{\prime }$ for each entry $\left( i,j\right)$.} and for
each $i\in N$, 
\begin{equation*}
I_{i}\left( N,M,t\right) \geq I_{i}\left( N,M,t^{\prime }\right) .
\end{equation*}


Second, 
the axiom that says that if the streams of an artist increase, user by user,
and the streams of the other artists remain the same, then the importance of
this artist can not decrease and the importance of the other artists can not
increase.\footnote{%
Feng and Ma (2026) consider a somewhat related axiom in which the increase
in an artist's streams is due to a redistribution of streams a user makes.
That is, an artist increases streams because another artist decreased
streams (by the same margin).} Formally,

\textbf{Artist monotonicity}. For each pair $\left( N,M,t\right) ,\left(
N,M,t^{\prime }\right) \in \mathcal{P}$ and $i\in N$ such that $t_{i.}\geq
t_{i.}^{\prime }$ and $t_{i^{\prime }.}=t_{i^{\prime }.}^{\prime }$ for all $%
i^{\prime }\in N\backslash \left\{ i\right\} ,$ 
\begin{eqnarray*}
I_{i}\left( N,M,t\right) &\geq &I_{i}\left( N,M,t^{\prime }\right) \text{ and%
} \\
I_{i^{\prime }}\left( N,M,t\right) &\leq &I_{i^{\prime }}\left(
N,M,t^{\prime }\right) \text{ for all }i^{\prime }\in N\backslash \left\{
i\right\} .
\end{eqnarray*}


Third, 
a similar but more demanding axiom than artist monotonicity. It says that if
the total streams of an artist increase and the total streams of the other
artists remain the same, then the importance of this artist can not decrease
and the importance of the other artists cannot increase. Formally,

\textbf{Total artist monotonicity.} For each pair $\left( N,M,t\right)
,\left( N,M,t^{\prime }\right) \in \mathcal{P}$ and $i\in N$ such that $%
T_{i}\geq T_{i}^{\prime }$ and $T_{i^{\prime }}=T_{i^{\prime }}^{\prime }$
for all $i^{\prime }\in N\backslash \left\{ i\right\} ,$ 
\begin{eqnarray*}
I_{i}\left( N,M,t\right) &\geq &I_{i}\left( N,M,t^{\prime }\right) \text{ and%
} \\
I_{i^{\prime }}\left( N,M,t\right) &\leq &I_{i^{\prime }}\left(
N,M,t^{\prime }\right) \text{ for all }i^{\prime }\in N\backslash \left\{
i\right\} .
\end{eqnarray*}


Suppose now that the streams of an artist remain the same whereas the
streams of the other artists increase. 
The next axiom says that the importance of this artist can not increase and
the importance of the other artists can not decrease\footnote{%
Berganti\~{n}os and Moreno-Ternero (2022) considers a similar axiom in
broadcasting problems.}. 
Formally,

\textbf{Reciprocal artist monotonicity. }For each pair $\left( N,M,t\right)
,\left( N,M,t^{\prime }\right) \in \mathcal{P}$ and each $i\in N$ such that $%
t_{i.}=t_{i.}^{\prime }$ and $t_{i^{\prime }.}\geq t_{i^{\prime }.}^{\prime
} $ for all $i^{\prime }\in N\backslash \left\{ i\right\} ,$ 
\begin{eqnarray*}
I_{i}\left( N,M,t\right) &\leq &I_{i}\left( N,M,t^{\prime }\right) \text{ and%
} \\
I_{i^{\prime }}\left( N,M,t\right) &\geq &I_{i^{\prime }}\left(
N,M,t^{\prime }\right) \text{ for all }i^{\prime }\in N\backslash \left\{
i\right\} .
\end{eqnarray*}

Finally, an axiom saying that when a user increases its streams, artist by
artist, and the other users remain the same, then the importance of all
artists streamed by the user can not decrease. Formally,

\textbf{User monotonicity}. For each pair $\left( N,M,t\right) ,\left(
N,M,t^{\prime }\right) \in \mathcal{P}$ and each $j\in M$ such that $%
t_{.j}\geq t_{.j}^{\prime }$ and $t_{.j^{\prime }}=t_{.j^{\prime }}^{\prime
} $ for all $j^{\prime }\in M\backslash \left\{ j\right\} ,$ 
\begin{equation*}
I_{i}\left( N,M,t\right) \geq I_{i}\left( N,M,t^{\prime }\right) \text{ for
all }i\in L^{j}\left( N,M,t\right) .
\end{equation*}





\begin{remark}
\label{rel mon axioms}We have the following obvious logical relations among
the monotonicity axioms.

$\left( a\right) $ Overall monotonicity implies user monotonicity.

$\left( b\right) $ Total artist monotonicity implies artist monotonicity.
\end{remark}


\subsection{Extra axioms}

We now consider some extra axioms. The first three axioms are standard in
the axiomatic literature and have been used in many kind of problems,
including streaming problems. The last two axioms are more specific of
streaming problems.

\bigskip

First, the axiom saying that if the streams of two artists coincide for each
user, then the index for both artists should coincide. Formally,

\textbf{Symmetry.} For each $\left( N,M,t\right) \in \mathcal{P}$ and each
pair $i,i^{\prime }\in N$ such that $t_{i.}=t_{i^{\prime }.}$, 
\begin{equation*}
I_{i}\left( N,M,t\right) =I_{i^{\prime }}\left( N,M,t\right) .
\end{equation*}

A strengthening of symmetry is the following. 
Let $\Pi \left( M,M^{\prime }\right) $ with $\left\vert M\right\vert
=\left\vert M^{\prime }\right\vert $ be the set of bijections from $M$ to $%
M^{\prime }.$ For each $\left( N,M,t\right) \in \mathcal{P}$ and each $\pi
\in $ $\Pi \left( M,M^{\prime }\right) $ we denote by $\left( N,M^{\prime
},t^{\pi }\right) $ the problem where for each $i\in N$ and $j\in M^{\prime
},$ $t_{ij}^{\pi }=t_{i\pi ^{-1}\left( j\right)}.$

\textbf{Anonymity on users}. For each $\left( N,M,t\right) \in \mathcal{P}$,
each $\pi \in \Pi \left( M,M^{\prime }\right) ,$ and each $i\in N,$ 
\begin{equation*}
I_{i}\left( N,M,t\right) =I_{i}\left( N,M^{\prime },t^{\pi }\right) .
\end{equation*}

Next, the axiom saying that the index divides the overall budget (or,
equivalently, that the index is normalized and thus coincides with its
associated sharing rule). Formally,

\textbf{Budget balance.} For each $\left( N,M,t\right) \in \mathcal{P},$ 
\begin{equation*}
\sum_{i\in N}I_{i}\left( N,M,t\right) =m.
\end{equation*}

The next axiom, introduced in Berganti\~{n}os and Moreno-Ternero (2025b),
says that if each user has played artist $i$ a certain times more than
artist $i^{\prime },$ then, the index should preserve that ratio. Formally,

\textbf{Pairwise homogeneity.} For each $\left( N,M,t\right) \in \mathcal{P}$%
, each pair $i,k\in N,$ and each $\lambda \geq 0$ such that $t_{ij}=\lambda
t_{kj}$ for all $j\in M,$ $\ $ 
\begin{equation*}
I_{i}\left( N,M,t\right) =\lambda I_{k}\left( N,M,t\right) .
\end{equation*}

We consider a final axiom introduced in Berganti\~{n}os and Moreno-Ternero
(2025b), which says that an artist cannot increase the index unilaterally
upon increasing its number of streams from fans.

\textbf{Non-unilateral manipulability}. For each pair $\left(
N,M,t\right),\left( N,M,t^{\prime }\right) \in \mathcal{P}$, and each $i\in
N $ such that $F_{i}\left( N,M,t\right) =F_{i}\left( N,M,t^{\prime }\right) $%
, $t_{ij}\leq t_{ij}^{\prime }$ for all $j\in M$, and $t_{kj}=t_{kj}^{\prime
}$ for all $k\in N\backslash \left\{ i\right\} $ and $j\in M$, 
\begin{equation*}
I_{i}\left( N,M,t\right) \geq I_{i}\left( N,M,t^{\prime }\right) .
\end{equation*}

\section{Characterizations of specific indices}


We gather in the next result how each of the four indices perform with
respect to the monotonicity axioms, which will pave the way for the
characterization results we present afterwards.


\begin{proposition}
\label{axioms indices} The following statements hold:

$\left( a\right) $ The pro-rata and egalitarian index satisfy all the
monotonicity axioms.

$\left( b\right) $ The user-centric and Shapley index satisfy artist
monotonicity but violate overall monotonicity, total artist monotonicity,
reciprocal artist monotonicity and user monotonicity.


\end{proposition}

\bigskip 

\begin{proof}
We concentrate on the statements on violations as the remaining ones are
straightforward to check. Let us start with the user-centric index. To show
that it does not satisfy \textit{total artist monotonicity}, let $N=\left\{
1,2,3\right\} ,$ $M=\left\{ 1,2\right\} ,$ 
\begin{equation*}
t=\left( 
\begin{array}{cc}
0 & 30 \\ 
10 & 0 \\ 
0 & 70%
\end{array}%
\right) ,\text{ and }t^{\prime }=\left( 
\begin{array}{cc}
10 & 0 \\ 
10 & 0 \\ 
0 & 70%
\end{array}%
\right) .
\end{equation*}%
Then, 
\begin{equation*}
U_{1}\left( N,M,t\right) =0.3<0.5=U_{1}\left( N,M,t^{\prime }\right) .
\end{equation*}%
As for \textit{reciprocal artist monotonicity}, let $N=\left\{ 1,2,3\right\}
,$ $M=\left\{ 1,2\right\} ,$ 
\begin{equation*}
t=\left( 
\begin{array}{cc}
10 & 0 \\ 
0 & 20 \\ 
0 & 30%
\end{array}%
\right) ,\text{ and }t^{\prime }=\left( 
\begin{array}{cc}
10 & 0 \\ 
0 & 10 \\ 
0 & 10%
\end{array}%
\right) .
\end{equation*}%
Then, 
\begin{equation*}
U_{2}\left( N,M,t\right) =0.4<0.5=U_{2}\left( N,M,t^{\prime }\right) .
\end{equation*}%
As for \textit{user monotonicity}, let $N=\left\{ 1,2,3\right\} ,$ $%
M=\left\{ 1\right\} ,$ $t=\left( 10,90,0\right) $ and $t^{\prime }=\left(
5,5,0\right) $. Then, 
\begin{equation*}
U_{1}\left( N,M,t\right) =0.1<0.5=U_{1}\left( N,M,t^{\prime }\right) .
\end{equation*}%
Finally, the above also implies that it does not satisfy \textit{overall
monotonicity} either (recall Remark \ref{rel mon axioms}).

We now turn to the Shapley index. To show that it does not satisfy \textit{%
total artist monotonicity}, let $N=\left\{ 1,2,3\right\} ,$ $M=\left\{
1,2\right\} ,$ 
\begin{equation*}
t=\left( 
\begin{array}{cc}
0 & 30 \\ 
10 & 0 \\ 
0 & 70%
\end{array}%
\right) ,\text{ and }t^{\prime }=\left( 
\begin{array}{cc}
10 & 10 \\ 
10 & 0 \\ 
0 & 70%
\end{array}%
\right) .
\end{equation*}%
Then, 
\begin{equation*}
Sh_{1}\left( N,M,t\right) =0.5<1=Sh_{1}\left( N,M,t^{\prime }\right) .
\end{equation*}%
As for \textit{reciprocal artist monotonicity}, let $N=\left\{ 1,2,3\right\}
,$ $M=\left\{ 1,2\right\} ,$ 
\begin{equation*}
t=\left( 
\begin{array}{cc}
10 & 0 \\ 
20 & 20 \\ 
20 & 20%
\end{array}%
\right) ,\text{ and }t^{\prime }=\left( 
\begin{array}{cc}
10 & 0 \\ 
20 & 10 \\ 
20 & 0%
\end{array}%
\right) .
\end{equation*}%
Then, 
\begin{equation*}
Sh_{2}\left( N,M,t\right) =0.83<1.33=Sh_{2}\left( N,M,t^{\prime }\right) .
\end{equation*}%
As for \textit{user monotonicity}, let $N=\left\{ 1,2,3\right\} ,$ $%
M=\left\{ 1\right\} ,$ $t=\left( 10,20,0\right) $ and $t^{\prime }=\left(
5,0,0\right) $. Then, 
\begin{equation*}
Sh_{1}\left( N,M,t\right) =\frac{1}{2}<1=Sh_{1}\left( N,M,t^{\prime }\right)
.
\end{equation*}%
Finally, the above implies that it does not satisfy \textit{overall
monotonicity} either (recall Remark \ref{rel mon axioms}).  
\end{proof}

\bigskip 

\bigskip 

\bigskip 

The content of Proposition \ref{axioms indices} can be summarized in the
following table. 
\begin{equation*}
\begin{tabular}{|l|l|l|l|l|}
\hline
Axiom / Index & Pro-rata & User-centric & Shapley & Egalitarian \\ \hline
Overall monotonicity & Yes & No & No & Yes \\ \hline
Artists monotonicity & Yes & Yes & Yes & Yes \\ \hline
Total artist monotonicity & Yes & No & No & Yes \\ \hline
Reciprocal artist monotonicity & Yes & No & No & Yes \\ \hline
User monotonicity & Yes & No & No & Yes \\ \hline
\end{tabular}%
\end{equation*}

\bigskip


\subsection{The pro-rata index}

We first provide two characterizations of the pro-rata index resorting to
two monotonicity axioms.

\begin{theorem}
\label{char pro-rata} An index satisfies additivity, pairwise homogeneity,

$\left( a\right) $ reciprocal artist monotonicity and anonymity on users if
and only if it is the pro-rata index.

$\left( b\right) $ and total artist monotonicity if and only if it is the
pro-rata index.
\end{theorem}

\begin{proof}
$\left( a\right) $ Berganti\~{n}os and Moreno-Ternero (2025a) prove that
pro-rata satisfies additivity and pairwise homogeneity. By Proposition \ref%
{axioms indices} pro-rata satisfies reciprocal artists monotonicity. It is
obvious that pro-rata satisfies anonymity on users.

Let $I$ be an index satisfying \textit{additivity}, \textit{pairwise
homogeneity}, \textit{reciprocal artist monotonicity} and \textit{anonymity
on users}. 
By Theorem 1 in Berganti\~{n}os and Moreno-Ternero (2025a), \textit{%
additivity} and \textit{pairwise homogeneity} imply that $I$ is a \textit{%
weighted index}. That is, there exists a function $w:\mathbb{M}\times 
\mathbb{M}\times \bigcup\limits_{N\subset \mathbb{N}}\mathbb{Z}%
_{+}^{n}\rightarrow \mathbb{R}_{+}$ such that for each $\left( N,M,t\right)
\in \mathcal{P}$ and each $i\in N,$ 
\begin{equation*}
I_{i}^{w}\left( N,M,t\right) =\sum_{j\in M}w\left( j,t_{.j}\right) t_{ij}.
\end{equation*}

We show next that \textit{reciprocal artist monotonicity} implies that it is
actually a \textit{user-weighted index}, i.e., we show that, for each $%
\left( N,M,t\right) \in \mathcal{P}$, and each $i\in N$, there exists $w:%
\mathbb{M}\rightarrow \mathbb{R}_{+}$ such that 
\begin{equation*}
I_{i}\left( N,M,t\right) =\sum_{j\in M}w\left( j\right) t_{ij}.
\end{equation*}

To do so, let $j\in M,$ $x\in \mathbb{Z}_{+}^{n},$ $k\in N,$ and $x^{k\alpha
}$ be such that $x_{k}^{k\alpha }=\alpha >x_{k}$ and $x_{i}^{k\alpha }=x_{i}$
for all $i\in N\backslash \left\{ k\right\} $. Let $k,i^{1},i^{2}$ be three
different artists with $x_{i^{1}}>0.$ By \textit{reciprocal artist
monotonicity} applied to $\left( N,\left\{ j\right\} ,x\right) ,$ $\left(
N,\left\{ j\right\} ,x^{k\alpha }\right) $ and $i^{1}$ we obtain that 
\begin{equation*}
w\left( j,x\right) x_{i^{1}}=I_{i^{1}}\left( N,\left\{ j\right\} ,x\right)
\leq I_{i^{1}}\left( N,\left\{ j\right\} ,x^{k\alpha }\right) =w\left(
j,x^{k\alpha }\right) x_{i^{1}}.
\end{equation*}%
Thus, $w\left( j,x\right) \leq w\left( j,x^{k\alpha }\right) $

By \textit{reciprocal artist monotonicity} applied to $\left( N,\left\{
j\right\} ,x\right) ,$ $\left( N,\left\{ j\right\} ,x^{k\alpha }\right) $
and $i^{2}$ we obtain that 
\begin{equation*}
w\left( j,x\right) x_{i^{1}}=I_{i^{1}}\left( N,\left\{ j\right\} ,x\right)
\geq I_{i^{1}}\left( N,\left\{ j\right\} ,x^{k\alpha }\right) =w\left(
j,x^{k\alpha }\right) x_{i^{1}}.
\end{equation*}%
Thus, $w\left( j,x\right) \geq w\left( j,x^{k\alpha }\right) .$ Hence $%
w\left( j,x\right) =w\left( j,x^{k\alpha }\right) .$

Let $\alpha >\max \left\{ x_{i}:i\in N\right\} .$ Then, 
\begin{equation*}
w\left( j,x\right) =w\left( j,x^{1\alpha }\right) =w\left( j,\left(
x^{1\alpha }\right) ^{2\alpha }\right) =...=w\left( j,\left( \left(
x^{1\alpha }\right) ^{2\alpha }\right) ^{...n\alpha }\right) =w\left(
j,\left( \alpha ,\alpha ,...,\alpha \right) \right) .
\end{equation*}

Now, given $j\in \mathbb{M}$ and $x,x^{\prime }\in \mathbb{Z}_{+}^{n}$ we
define $\alpha >\max \left\{ \max \left\{ x_{i}:i\in N\right\} ,\max \left\{
x_{i}^{\prime }:i\in N\right\} \right\} .$ Thus, 
\begin{equation*}
w\left( j,x\right) =w\left( j,\left( \alpha ,\alpha ,...,\alpha \right)
\right) =w\left( j,x^{\prime }\right)
\end{equation*}%
and hence $I$ is a \textit{user-weighted index}.

Note that, given $j\in M$ and $i\in N$, $w\left( j\right) =I_{i}\left(
N,\left\{ j\right\} ,\left( 1,...,1\right) \right) .$ Let $j^{\prime }\in M.$
By \textit{anonymity on users}, $I_{i}\left( N,\left\{ j^{\prime }\right\}
,\left( 1,...,1\right) \right) =I_{i}\left( N,\left\{ j\right\} ,\left(
1,...,1\right) \right) .$ Thus, $w\left( j\right) =w\left( j^{\prime
}\right) .$ Let us define $w=w\left( j\right) $ with $j\in M.$ Then, for
each $\left( N,M,t\right) \in \mathcal{P}\ $and each $i\in N$, 
\begin{equation*}
I_{i}\left( N,M,t\right) =\sum_{j\in M}wt_{ij}=wP\left( N,M,t\right) .
\end{equation*}

As we have identified an index with all its positive linear transformations,
the statement follows. 

$\left( b\right) $ Berganti\~{n}os and Moreno-Ternero (2025a) prove that
pro-rata satisfies additivity and pairwise homogeneity. By Proposition \ref%
{axioms indices} pro-rata satisfies total artist monotonicity.

Let $I$ be an index satisfying \textit{additivity}, \textit{pairwise
homogeneity} and \textit{total artist monotonicity}. Berganti\~{n}os and
Moreno-Ternero (2025a) prove that if an index $I$ satisfies \textit{%
additivity} and \textit{pairwise homogeneity}, then $I$ is a \textit{%
weighted index} $I^{w}$.

Let $\left( N,M,t\right) ,\left( N,M,t^{\prime }\right) \in \mathcal{P}$ be
such that $T_{i}=T_{i}^{\prime }$ for all $i\in N.$ By \textit{total artists
monotonicity}, $I\left( N,M,t\right) =I\left( N,M,t^{\prime }\right) .$

Let $\left( N,M,t\right) \in \mathcal{P}$ and $i\in N.$ We define $\left(
N,M,t^{i,1}\right) $ where 
\begin{equation*}
t_{kl}^{i,1}=\left\{ 
\begin{tabular}{ll}
$T_{i}$ & if $k=i$ and $l=1$ \\ 
0 & if $k=i$ and $l\neq 1$ \\ 
$0$ & if $k\neq i$ and $l=1$ \\ 
$t_{kl}+t_{k1}$ & if $k\neq i$ and $l=2$ \\ 
$t_{kl}$ & otherwise.%
\end{tabular}%
\right.
\end{equation*}

Let $1^{i}$ be the vector where $1_{i}^{i}=1$ and $1_{k}^{i}=0$ when $k\in
N\backslash \left\{ i\right\}$. By \textit{total artist monotonicity}, 
\begin{equation*}
I_{i}\left( N,M,t\right) =I_{i}\left( N,M,t^{i,1}\right) =w\left(
1,T_{i}1^{i}\right) T_{i}.
\end{equation*}

Similarly, we can prove that $I_{i}\left( N,M,t\right) =w\left(
j,T_{i}1^{i}\right) T_{i}$ for all $j\in M.$

Let $i,i^{\prime }\in N,$ $j,j^{\prime }\in M,$ and $x\in \mathbb{Z}_{+}.$
We define $t^{1}$, $t^{2},$ and $t^{3}$ as follows: 
\begin{eqnarray*}
t_{kl}^{1} &=&\left\{ 
\begin{tabular}{ll}
$x$ & if $k=i$ and $l=j$ \\ 
$0$ & otherwise%
\end{tabular}%
\right. , \\
t_{kl}^{2} &=&\left\{ 
\begin{tabular}{ll}
$x$ & if $k=i^{\prime }$ and $l=j^{\prime }$ \\ 
$0$ & otherwise%
\end{tabular}%
\right. ,\text{ and} \\
t_{kl}^{3} &=&\left\{ 
\begin{tabular}{ll}
$x$ & if $k\in \left\{ i,i^{\prime }\right\} $ and $l=j$ \\ 
$0$ & otherwise%
\end{tabular}%
\right. ,\text{ }
\end{eqnarray*}

Thus, $I_{i}\left( N,M,t^{1}\right) =w\left( j,x1^{i}\right) x,$ $%
I_{i^{\prime }}\left( N,M,t^{2}\right) =w\left( j^{\prime },x1^{i^{\prime
}}\right) x,$ and $I_{i}\left( N,M,t^{3}\right) =w\left( j,t_{.j}^{3}\right)
x=I_{i^{\prime }}\left( N,M,t^{3}\right) .$

By \textit{total artist monotonicity}, $I_{i}\left( N,M,t^{1}\right)
=I_{i}\left( N,M,t^{3}\right) $ and $I_{i^{\prime }}\left( N,M,t^{2}\right)
=I_{i^{\prime }}\left( N,M,t^{3}\right) .$ Then, $w\left( j,x1^{i}\right)
=w\left( j^{\prime },x1^{i^{\prime }}\right) .$

Thus, for each $x\in \mathbb{Z}_{+}$ we can define $w\left( x\right) $ as $%
w\left( j,x1^{i}\right) $ because it does not depends on $i$ or $j.$ We now
prove that $w\left( 2\right) =w\left( 1\right) .$ Let $\left( N,M,t\right) ,$
$i\in N,$ and $j,j^{\prime }\in M.$ We now define $t^{1}$ and $t^{2}$ as
follows 
\begin{equation*}
t_{kl}^{1}=\left\{ 
\begin{tabular}{ll}
$2$ & if $k=i$ and $l=j$ \\ 
$0$ & otherwise%
\end{tabular}%
\right. \text{ and }t_{kl}^{2}=\left\{ 
\begin{tabular}{ll}
$1$ & if $k=i$ and $l\in \left\{ j,j^{\prime }\right\} $ \\ 
$0$ & otherwise%
\end{tabular}%
\right. .
\end{equation*}

Thus, $I_{i}\left( N,M,t^{1}\right) =w\left( 2\right) 2$ and $I_{i}\left(
N,M,t^{2}\right) =w\left( 1\right) +w\left( 1\right) .$ By \textit{total
artists monotonicity}, $I_{i}\left( N,M,t^{1}\right) =I_{i}\left(
N,M,t^{2}\right) .$ Hence, $w\left( 2\right) =w\left( 1\right) .$

Assume now that $w\left( x\right) =w\left( 1\right) $ for all $x\leq a.$ We
prove it when $x=a+1.$ Let $\left( N,M,t\right) ,$ $i\in N,$ and $%
j,j^{\prime }\in M.$ We now define $t^{1}$ and $t^{2}$ as follows 
\begin{equation*}
t_{kl}^{1}=\left\{ 
\begin{tabular}{ll}
$a+1$ & if $k=i$ and $l=j$ \\ 
$0$ & otherwise%
\end{tabular}%
\right. \text{ and }t_{kl}^{2}=\left\{ 
\begin{tabular}{ll}
$a$ & if $k=i$ and $l=j$ \\ 
1 & if $k=i$ and $l=j^{\prime }$ \\ 
$0$ & otherwise%
\end{tabular}%
\right. .
\end{equation*}

Thus, $I_{i}\left( N,M,t^{1}\right) =w\left( a+1\right) \left( a+1\right) $
and $I_{i}\left( N,M,t^{2}\right) =w\left( a\right) a+w\left( 1\right) .$ By 
\textit{total artist monotonicity}, $I_{i}\left( N,M,t^{1}\right)
=I_{i}\left( N,M,t^{2}\right) .$ Hence, $w\left( a+1\right) \left(
a+1\right) =w\left( a\right) a+w\left( 1\right) .$ By the induction
hypothesis, $w\left( a\right) =w\left( 1\right) .$ Hence, $w\left(
a+1\right) =w\left( 1\right) .$

Thus, for each $\left( N,M,t\right) \in \mathcal{P}\ $and each $i\in N$, 
\begin{equation*}
I_{i}\left( N,M,t\right) =\sum_{j\in M}w\left( 1\right) t_{ij}=w\left(
1\right) T_{i}=w\left( 1\right) P_{i}\left( N,M,t\right) .
\end{equation*}

As we have identified each index with all its positive linear
transformations, we deduce that $I$ coincides with $P$.
\end{proof}


\bigskip 

\begin{remark}
\label{indep char pro-rata} The axioms used in Theorem \ref{char pro-rata}
are independent.

$I^{2}$ defined above satisfies all axioms except for additivity.

The egalitarian index satisfies all axioms except for pairwise homogeneity.

Let $I$ be a weighted index that it is not a user-weighted index. Then, $I$
satisfies all axioms except for reciprocal artist monotonicity and total
artist monotonicity.

Let $I^{w}$ be a weighted index where for each $(j,x)\in \mathbb{M}\times
\bigcup\limits_{N\subset \mathbb{N}}\mathbb{Z}^n$, $w\left( j,x\right) =j.$
Then, $I^{w}$ satisfies all axioms except for anonymity on users.
\end{remark}

\bigskip 

\subsection{The Shapley index}

We now provide a characterization of the Shapley index resorting to another
monotonicity axiom.

\begin{theorem}
\label{char Sh} An index satisfies additivity, symmetry, budget balance,
null artist, non-unilateral manipulability and artist monotonicity if and
only if it is the Shapley index.
\end{theorem}

\begin{proof}
Berganti\~{n}os and Moreno-Ternero (2025b) prove that the Shapley index
satisfies all the axioms in the statement but \textit{artist monotonicity}.
By Proposition \ref{axioms indices}, the Shapley index also satisfies 
\textit{artist monotonicity}.

Conversely, let $I$ be an index satisfying the axioms of the statement.

Let $\left( N,M,t\right) \in \mathcal{P}$. By \textit{additivity}, it is
enough to prove that $I$ coincides with the Shapley index at $\left(
N,\left\{ j\right\} ,t_{.j}\right)$, for all $j\in M.$

Given $j\in M$, we define $S=\{i\in N:t_{ij}=0\}$. Assume, without loss of
generality that $N\backslash S=\left\{ 1,...,n-s\right\} .$

By \textit{null artist}, for each $i\in S,$ $I_{i}\left( N,\left\{ j\right\}
,t_{.j}\right) =0$.

Let $t^{\prime }$ be such that $t_{ij}^{\prime }=0$ if $i\in S$ and $%
t_{ij}^{\prime }=\max \{t_{ij}:i\in N\}$ otherwise. By \textit{null artist},
for each $i\in S$, $I_{i}\left( N,\left\{ j\right\} ,t^{\prime }\right) =0$.
By \textit{symmetry}, $I_{i}\left( N,\left\{ j\right\} ,t^{\prime }\right)
=x $, for each $i\notin S$. By \textit{budget balance}, $x=\frac{1}{n-s}$.
Thus, $I\left( N,\left\{ j\right\} ,t^{\prime }\right) =Sh\left( N,\left\{
j\right\} ,t^{\prime }\right) .$

For each $k=1,...,n-s$ we define $x^{k}$ such that $x_{i}^{k}=\max
\{t_{ij}:i\in N\}$ when $i\leq k$ and $x_{i}^{k}=t_{ij}$ otherwise. We
define $x^{0}=t_{.j}.$

Let $k\in \left\{ 1,...,n-s\right\} .$ By \textit{artist monotonicity}, 
\begin{eqnarray*}
I_{k}\left( N,\left\{ j\right\} ,x^{k}\right) &\geq &I_{k}\left( N,\left\{
j\right\} ,x^{k-1}\right) \text{ and } \\
I_{i}\left( N,\left\{ j\right\} ,x^{k}\right) &\leq &I_{i}\left( N,\left\{
j\right\} ,x^{k-1}\right) \text{ for each }i\in N\backslash \left\{
k\right\} .
\end{eqnarray*}

By \textit{non-unilateral manipulability}, 
\begin{equation*}
I_{k}\left( N,\left\{ j\right\} ,x^{k}\right) \leq I_{k}\left( N,\left\{
j\right\} ,x^{k-1}\right) .
\end{equation*}

Thus, 
\begin{equation*}
I_{k}\left( N,\left\{ j\right\} ,x^{k}\right) =I_{k}\left( N,\left\{
j\right\} ,x^{k-1}\right) .
\end{equation*}

By \textit{budget balance}, for each $i\in N\backslash \left\{ k\right\} $ 
\begin{equation*}
I_{i}\left( N,\left\{ j\right\} ,x^{k}\right) =I_{i}\left( N,\left\{
j\right\} ,x^{k-1}\right) .
\end{equation*}

Then, $I\left( N,\left\{ j\right\} ,x^{k}\right) =I\left( N,\left\{
j\right\} ,x^{k-1}\right) .$ As $x^{0}=t_{.j}$ and $x^{n-s}=t_{.j}^{\prime }$
it follows that 
\begin{equation*}
I\left( N,\left\{ j\right\} ,t_{.j}\right) =I\left( N,\left\{ j\right\}
,t_{.j}^{\prime }\right) =Sh\left( N,\left\{ j\right\} ,t^{\prime }\right)
=Sh\left( N,\left\{ j\right\} ,t\right),
\end{equation*}
which concludes the proof.
\end{proof}


\bigskip 

\begin{remark}
\label{indep char Sh} The axioms used in Theorem \ref{char Sh} are
independent.

We consider the index that divides $m$ equally among non-null artists.
Formally, given a problem $\left( N,M,t\right) $ we define $N^{nn}$ as the
set of non null artists. Namely, $N^{nn}=\left\{ i\in N:T_{i}>0\right\} .$
Let $I$ be such that $I_{i}\left( N,M,t\right) =\frac{m}{\left\vert
N^{nn}\right\vert }$ if $i\in N^{nn}$ and $I_{i}\left( N,M,t\right) =0$
otherwise. This index satisfies all axioms except for additivity.

$I^{w}$ defined above satisfies all axioms except for symmetry.

The index that assigns 0 to each artist in each problem satisfies all axioms
except for budget balance.

The egalitarian index satisfies all axioms except for null artist.

The user-centric index satisfies all axioms except for non-unilateral
manipulability

Consider the index $U^{-1}$ defined as follows. For each $j\in M$ 
\begin{equation*}
U_{i}^{-1}\left( N,\left\{ j\right\} ,t\right) =\left\{ 
\begin{tabular}{cc}
0 & if $t_{ij}=0$ \\ 
$\frac{\frac{1}{t_{ij}}}{\sum\limits_{i^{\prime }:t_{i^{\prime }j}>0}\frac{1%
}{t_{i^{\prime }j}}}$ & if $t_{i^{\prime }j}>0.$%
\end{tabular}%
\right. 
\end{equation*}%
We extend $U^{-1}$ to $\left( N,M,t\right) $ by additivity. Then $U^{-1}$
satisfies all axioms except for artist monotonicity.
\end{remark}

\bigskip 

\bigskip 

\bigskip 

\subsection{The egalitarian index}

We conclude showing how some of the monotonicity axioms are somewhat strong
as they characterize the egalitarian index, when combined with two basic
axioms.


\begin{theorem}
\label{char egalitarian} An index satisfies symmetry, budget balance, and

$\left( a\right) $ monotonicity if and only if it is the egalitarian index.

$\left( b\right) $  user monotonicity if and only if it is the egalitarian
index.

$\left( c\right) $ reciprocal artist monotonicity if and only if it is the
egalitarian index.
\end{theorem}

\bigskip

\begin{proof}
It is straightforward to show that the egalitarian index satisfies all the
axioms in the statements. Conversely, let $I$ be an index satisfying \textit{%
symmetry} and \textit{budget balance}. To show statement (b) assume that it
also satisfies \textit{user monotonicity} (by Remark \ref{rel mon axioms},
this proof would also be valid to prove statement (a)).

Let $\left( N,M,t\right) \in \mathcal{P}$. We define $a>\max \left\{
t_{ij}:i\in N,j\in M\right\} .$

Assume that $M=\left\{ j^{1},...,j^{m}\right\} .$ For each $k=0,1,...,m$ we
define $t^{k}$ as follows, 
\begin{equation*}
t_{ij^{l}}^{k}=\left\{ 
\begin{tabular}{cc}
$a$ & if $l\leq k$ \\ 
$t_{ij^{l}}$ & otherwise.%
\end{tabular}%
\right.
\end{equation*}

Notice that $t=t^{0}\leq t^{1}\leq ...\leq t^{m}$ and $t_{ij}^{m}=a$ for
each $i\in N$ and $j\in M.$ Besides, for each $k=1,...,m$ we have that $%
t_{.j^{k}}^{k}\geq t_{.j^{k}}^{k},$ $t_{.j}^{k}=t_{.j}^{k}$ for each $j\in
M\backslash \left\{ j^{k}\right\} ,$ and $L^{j^{k}}\left(
N,M,t^{j^{k}}\right) =N.$ By \textit{user monotonicity}, for each $i\in N,$ $%
I_{i}\left( N,M,t^{j^{k}}\right) \geq I_{i}\left( N,M,t^{j^{k-1}}\right) .$
Then, for each $i\in N,$ $I_{i}\left( N,M,t^{j^{m}}\right) \geq I_{i}\left(
N,M,t\right) .$

By \textit{budget balance}, for each $i\in N,$ $I_{i}\left(
N,M,t^{j^{m}}\right) =I_{i}\left( N,M,t\right) .$

By \textit{symmetry}, for each $i\in N,$ $I_{i}\left( N,M,t^{j^{m}}\right) =%
\frac{m}{n}.$ Thus, for each $i\in N,$ 
\begin{equation*}
I_{i}\left( N,M,t\right) =\frac{m}{n}=E_{i}\left( N,M,t\right),
\end{equation*}
as desired.

We now move to statement $\left( c\right)$. 
That is, assume now that $I$ satisfies \textit{reciprocal artist monotonicity%
} instead.

Let $\left( N,M,t\right) \in \mathcal{P}$. We define $a$ as above. We order
the pairs $\left\{ \left( i,j\right) \right\} _{i\in N,j\in M}$
lexicographically. Namely, 
\begin{equation*}
\left( 1,1\right) \leq _{L}\left( 1,2\right) \leq _{L}...\leq _{L}\left(
1,m\right) \leq _{L}\left( 2,1\right) \leq _{L}\left( 2,2\right) ....
\end{equation*}

For each $k=1,...,nm$ we define $\left( i^{k},j^{k}\right) $ as the
corresponding pair in the previous order. Besides, we define $t^{k}$ as
follows, 
\begin{equation*}
t_{ij}^{k}=\left\{ 
\begin{tabular}{cc}
$a$ & if $\left( i,j\right) \leq _{L}\left( i^{k},j^{k}\right) $ \\ 
$t_{ij}$ & otherwise.%
\end{tabular}%
\right.
\end{equation*}

Notice that $t=t^{0}\leq t^{1}\leq ...\leq t^{nm}$ and $t_{ij}^{nm}=a$ for
each $i\in N$ and each $j\in M$.

By \textit{symmetry}, for each $i,i^{\prime }\in N,$ $I_{i}\left(
N,M,t^{nm}\right) =I_{i^{\prime }}\left( N,M,t^{nm}\right) .$ By \textit{%
budget balance}, for each $i\in N,$ $I_{i}\left( N,M,t^{nm}\right) =\frac{m}{%
n}=E_{i}\left( N,M,t^{nm}\right) .$

We assume that for each $k=k^{\prime },...,nm$, $I\left( N,M,t^{k}\right)
=E\left( N,M,t^{k}\right) .$ We now prove that $I\left( N,M,t^{k^{\prime
}-1}\right) =E\left( N,M,t^{k^{\prime }-1}\right) .$ Notice that $%
t^{k^{\prime }}$ is obtained from $t^{k^{\prime }-1}$ by increasing entry $%
\left( i^{k^{\prime }},j^{k^{\prime }}\right) $ from $t_{i^{k^{\prime
}}j^{k^{\prime }}}$ to $a.$

Let $i,i^{\prime }\in N\setminus \{i^{k^{\prime }}\}$. By \textit{reciprocal
artists monotonicity} applied to $t^{k^{\prime }-1},$ $t^{k^{\prime }},$ and 
$i^{\prime },$ 
\begin{equation*}
I_{i}\left( N,M,t^{k^{\prime }}\right) \geq I_{i}\left( N,M,t^{k^{\prime
}-1}\right) \text{.}
\end{equation*}

By \textit{reciprocal artist monotonicity} applied to $t^{k^{\prime }-1},$ $%
t^{k^{\prime }},$ and $i,$ 
\begin{equation*}
I_{i}\left( N,M,t^{k^{\prime }}\right) \leq I_{i}\left( N,M,t^{k^{\prime
}-1}\right) \text{.}
\end{equation*}

Then, 
\begin{equation*}
I_{i}\left( N,M,t^{k^{\prime }-1}\right) =I_{i}\left( N,M,t^{k^{\prime
}}\right) =E_{i}\left( N,M,t^{k^{\prime }}\right) =\frac{m}{n}.
\end{equation*}

As $I_{i}\left( N,M,t^{k^{\prime }-1}\right) =\frac{1}{n}$ for each $i\in
N\setminus \{i^{k^{\prime }}\}$, by \textit{budget balance}, 
\begin{equation*}
I_{i^{k^{\prime }}}\left( N,\left\{ j\right\} ,t\right) =\frac{m}{n}.
\end{equation*}

Thus, $I\left( N,M,t^{k^{\prime }-1}\right) =E\left( N,M,t^{k^{\prime
}-1}\right) .$

Consequently, $I\left( N,M,t^{0}\right) =E\left( N,M,t^{0}\right)$. As $%
t^{0}=t$, we conclude that $I\equiv E$, as desired.
\end{proof}


\bigskip 

\begin{remark}
\label{indep char egal} The axioms used in Theorem \ref{char egalitarian}
are independent.

Let $w=\left( w_{i}\right) _{i\in \mathbb{N}}$ be such that $w_{i}>0$ for
all $i\in \mathbb{N}$. We define the index $R^{w}$ as follows: for each
problem $\left( N,M,t\right) \in \mathcal{P}$ and each $i\in N$, 
\begin{equation*}
R_{i}^{w}\left( N,M,t\right) =\frac{w_{i}}{\sum\limits_{j\in N}w_{j}}m.
\end{equation*}

$R^{w}$ satisfies all axioms except for symmetry.

Pro-rata satisfies all axioms except for budget balance.

User-centric satisfies all axioms except for overall monotonicity/user
monotonicity/reciprocal artist monotonicity.
\end{remark}

\bigskip 

\bigskip 

\section{Characterizations of families of indices}

We complement our analysis in this section providing a final result that
gathers the effect monotonicity properties have on the family of
decomposable indices (the most general family in the literature). This
allows to uncover the structure of the space of indices, as well as the
effect of monotonicity axioms on this space.

\begin{theorem}
\label{char fam add+na} An index satisfies additivity, null artist and 

$\left( a\right) $ overall monotonicity or user monotonicity if and only if
it is a decomposable index satisfying the following condition: for each $%
i\in N\subset \mathbb{N}$, each $j\in \mathbb{M}$ and each pair $x,x^{\prime
}\in \mathbb{Z}_{+}^{n}$ 
such that $x\leq x^{\prime }$, $d\left( i,j,x\right) \leq d\left(
i,j,x^{\prime }\right) $.

$\left( b\right) $ artist monotonicity if and only if it is a decomposable
index satisfying the following condition: for each $i\in N\subset \mathbb{N}$%
, each $j\in \mathbb{M}$, each $x\in \mathbb{Z}_{+}^{n}$, and each $%
x_{i}^{\prime }\geq x_{i}$, 
$d\left( i,j,(x_{i}^{\prime },x_{-i})\right) \geq d\left(
i,j,(x_{i},x_{-i})\right) $ and $d\left( k,j,(x_{i},x_{-i})\right) \leq
d\left( k,j,(x_{i},x_{-i})\right) $, for each $k\in N\setminus \{i\}$.

$\left( c\right) $ total artist monotonicity if and only if it is a
decomposable index satisfying the following condition: for each $i\in
N\subset \mathbb{N}$, each $j\in M\subset \mathbb{M}$, each $t=\left(
t_{ij}\right) _{i\in N,j\in M},$ $t^{\prime }=\left( t_{ij}^{\prime }\right)
_{i\in N,j\in M}$ such that $T_{i}\geq T_{i}^{\prime }$ and $%
T_{k}=T_{k}^{\prime }$ for all $k\in N\backslash \left\{ i\right\} $, $%
\sum\limits_{j\in M}d\left( i,j,t_{.j}\right) \geq \sum\limits_{j\in
M}d\left( i,j,t_{.j}^{\prime }\right) $ and $\sum\limits_{j\in M}d\left(
k,j,t_{.j}\right) \geq \sum\limits_{j\in M}d\left( k,j,t_{.j}^{\prime
}\right) $ for each $k\in N\setminus \{i\}$.

$\left( d\right) $ reciprocal artist monotonicity if and only if it is a
decomposable index satisfying the following condition: for each $i\in
N\subset \mathbb{N}$, each $j\in \mathbb{M}$, each $x\in \mathbb{Z}_{+}^{n}$%
, and each $x_{-i}^{\prime }\geq x_{-i}$, 
$d\left( i,j,(x_{i},x_{-i})\right) \geq d\left( i,j,(x_{i},x_{-i}^{\prime
})\right) $ and $d\left( k,j,(x_{i},x_{-i})\right) \leq d\left(
k,j,(x_{i},x_{-i}^{\prime })\right) $, for each $k\in N\setminus \{i\}$.


\end{theorem}


\begin{proof}
As mentioned above, each decomposable index $I^{d}$ satisfies \textit{%
additivity} and \textit{null artist}. Thus, for the first direction of each
implication, we shall simply focus on its corresponding monotonicity axiom.

$\left( a\right)$ For this statement, we show first that, 
in the presence of \textit{additivity} and \textit{null artist}, \textit{%
user monotonicity} and \textit{overall monotonicity} are equivalent axioms.
In general, if an index satisfies \textit{overall monotonicity}, then it
also satisfies \textit{user monotonicity}. The converse also holds in the
presence of \textit{additivity} and \textit{null artist}, as we show next.
Let $\left( N,M,t\right) ,\left( N,M,t^{\prime }\right) \in \mathcal{P}$ be
such that $t\geq t^{\prime }$ and $i\in N$. By \textit{additivity}, it is
enough to show that for each $j\in M$, $I_{i}\left( N,\left\{ j\right\}
,t_{.j}\right) \geq I_{i}\left( N,\left\{ j\right\} ,t_{.j}^{\prime }\right)
.$ Let $j\in M.$ If $t_{.j}=t_{.j}^{\prime },$ then $I_{i}\left( N,\left\{
j\right\} ,t_{.j}\right) =I_{i}\left( N,\left\{ j\right\} ,t_{.j}^{\prime
}\right) .$ Assume that $t_{.j}\neq t_{.j}^{\prime }.$ By \textit{user
monotonicity}, $I_{i}\left( N,\left\{ j\right\} ,t_{.j}\right) \geq
I_{i}\left( N,\left\{ j\right\} ,t_{.j}^{\prime }\right) $ for each $i\in
L^{j}\left( N,\left\{ j\right\} ,t_{.j}\right) .$ Let $i\in N\backslash
L^{j}\left( N,\left\{ j\right\} ,t_{.j}\right) .$ Then, $t_{ij}=0=t_{ij}^{%
\prime }$. Thus, by \textit{null artist}, 
$I_{i}\left( N,\left\{ j\right\} ,t_{.j}\right) =0=I_{i}\left( N,\left\{
j\right\} ,t_{.j}^{\prime }\right)$. 
Altogether, $I_{i}\left( N,\left\{ j\right\} ,t_{.j}\right) \geq I_{i}\left(
N,\left\{ j\right\} ,t_{.j}^{\prime }\right)$, as desired.

We now show that the condition in the statement implies \textit{overall
monotonicity}. That is, let $\left(N,M,t\right) ,\left( N,M,t^{\prime
}\right) \in \mathcal{P}$ be such that $t\geq t^{\prime }$. Then, for each $%
j\in M$, $t_{.j}\ge t^{\prime }_{.j}$. Thus, by the condition in the
statement, for each $i\in N$, and each $j\in M$, $d(i,j,t_{.j})\geq
d(i,j,t_{.j}^{\prime })$. Consequently, 
\begin{equation*}
I_{i}^{d}\left( N,M,t\right) =\sum_{j\in M}d(i,j,t_{.j})\geq \sum_{j\in
M}d(i,j,t_{.j}^{\prime })=I_{i}^{d}\left( N,M,t^{\prime}\right) .
\end{equation*}%
Thus, $I^{d}$ satisfies \textit{overall monotonicity}.

Conversely, let $I$ be an index satisfying \textit{additivity}, \textit{null
artist} and \textit{overall monotonicity}. Then, as mentioned above, $I$ is
a decomposable index $I^{d}$. 
It just remains to show that $d$ satisfies the condition in the statement.
Let $i\in N\subset \mathbb{N}$, $j\in \mathbb{M}$ and $x,x^{\prime }\in 
\mathbb{Z}_{+}^{n}$ be such that $x\leq x^{\prime }$. By monotonicity, 
\begin{equation*}
d(i,j,x)=I_{i}^{d}\left( N,\{j\},x\right) \leq I_{i}^{d}\left(
N,\{j\},x^{\prime }\right) =d(i,j,x^{\prime }),
\end{equation*}%
as desired.

$\left( b\right) $ 
Let $\left( N,M,t\right) ,\left( N,M,t^{\prime }\right) \in \mathcal{P}$ and 
$i\in N$ be such that $t_{i.}\geq t_{i.}^{\prime }$ and $t_{k.}=t_{k.}^{%
\prime }$ for all $k\in N\backslash \left\{ i\right\} $. Let $j\in M$. By
the condition in the statement, 
\begin{equation*}
I_{i}^{d}\left( N,\left\{ j\right\} ,t\right) =d(i,j,t_{.j})\geq
d(i,j,t_{.j}^{\prime })=I_{i}^{d}\left( N,\left\{ j\right\} ,t^{\prime
}\right) .
\end{equation*}

Furthermore, for each $k\in N\backslash \left\{ i\right\}$, 
\begin{eqnarray*}
I_{k}^{d}\left( N,\left\{ j\right\} ,t\right) &=&d(k,j,t_{.j})\leq
d\left(k,j,t_{.j}^{\prime }\right)=I_{k}^{d}\left( N,\left\{
j\right\},t^{\prime }\right) \text{.}
\end{eqnarray*}

By \textit{additivity}, 
\begin{equation*}
I_{i}^{d}\left( N,M,t\right) =\sum_{j\in M}I_{i}^{d}\left( N,\left\{
j\right\} ,t_{.j}\right) \geq \sum_{j\in M}I_{i}^{d}\left(
N,\{j\},t_{.j}^{\prime }\right) =I_{i}^{d}\left( N,M,t^{\prime }\right) ,
\end{equation*}%
and for each $k\in N\backslash \left\{ i\right\} ,$ 
\begin{equation*}
I_{k}^{d}\left( N,M,t\right) =\sum_{j\in M}I_{k}^{d}\left( N,\left\{
j\right\} ,t_{.j}\right) \leq \sum_{j\in M}I_{k}^{d}\left(
N,\{j\},t_{.j}^{\prime }\right) =I_{k}^{d}\left( N,M,t^{\prime }\right) .
\end{equation*}%
Thus, $I^{d}$ satisfies \textit{artist monotonicity}.

Conversely, let $I$ be an index satisfying the three axioms. Then, as
mentioned above, $I$ is a decomposable index $I^{d}$. 
It just remains to show that $d$ satisfies the condition.

Let $i\in N\subset\mathbb{N}$, $j\in \mathbb{M}$, $x\in \mathbb{Z}_{+}^{n}$,
and $x^{\prime }_i\ge x_i$, 
Then, 
by \textit{artist monotonicity}, 
\begin{equation*}
d(i,j,(x^{\prime }_{i},x_{-i}))=I_{i}^{d}\left( N,\{j\},(x^{\prime
}_{i},x_{-i})\right) \geq I_{i}^{d}\left( N,\{j\},(x_{i},x_{-i})\right)
=d(i,j,(x_{i},x_{-i})).
\end{equation*}

Similarly, let $k\in N\backslash \left\{ i\right\} $. By \textit{artist
monotonicity}, 
\begin{equation*}
d(k,j,(x_{i}^{\prime },x_{-i}))=I_{k}^{d}\left( N,\{j\},(x_{i}^{\prime
},x_{-i})\right) \leq I_{k}^{d}\left( N,\{j\},(x_{i},x_{-i})\right)
=d(k,j,(x_{i},x_{-i})),
\end{equation*}%
as desired.

$\left( c\right)$ Let $\left( N,M,t\right) ,\left( N,M,t^{\prime }\right)
\in \mathcal{P}$ and $i\in N$ be such that $T_{i}\geq T_{i}^{\prime }$ and $%
T_{k}=T_{k}^{\prime }$ for all $k\in N\backslash \left\{ i\right\} $. By 
\textit{additivity}, 
\begin{eqnarray*}
I_{i}^{d}\left( N,M,t\right) &=&\sum_{j\in M}I_{i}^{d}\left( N,\left\{
j\right\} ,t_{.j}\right) =\sum_{j\in M}d\left( i,j,t_{.j}\right) \geq
\sum_{j\in M}d\left( i,j,t_{.j}^{\prime }\right) \\
&=&\sum_{j\in M}I_{i}^{d}\left( N,\{j\},t_{.j}^{\prime }\right)
=I_{i}^{d}\left( N,M,t^{\prime }\right) ,
\end{eqnarray*}%
and for each $k\in N\backslash \left\{ i\right\} ,$ 
\begin{eqnarray*}
I_{k}^{d}\left( N,M,t\right) &=&\sum_{j\in M}I_{k}^{d}\left( N,\left\{
j\right\} ,t_{.j}\right) =\sum_{j\in M}d\left( k,j,t_{.j}\right) \leq
\sum_{j\in M}d\left( k,j,t_{.j}^{\prime }\right) \\
&=&\sum_{j\in M}I_{k}^{d}\left( N,\{j\},t_{.j}^{\prime }\right)
=I_{k}^{d}\left( N,M,t^{\prime }\right) .
\end{eqnarray*}%
Thus, $I^{d}$ satisfies \textit{total artist monotonicity}.

Conversely, let $I$ be an index satisfying the three axioms. Then, as
mentioned above, $I$ is a decomposable index $I^{d}$. 
It just remains to show that $d$ satisfies the condition in the statement.
Let $i\in N\subset \mathbb{N}$, $j\in M\subset \mathbb{M}$, $t=\left(
t_{ij}\right) _{i\in N,j\in M},$ $t^{\prime }=\left( t_{ij}^{\prime }\right)
_{i\in N,j\in M}$ such that $T_{i}\geq T_{i}^{\prime }$ and $%
T_{k}=T_{k}^{\prime }$ for all $k\in N\backslash \left\{ i\right\} $, Then, 
by \textit{total artist monotonicity}, 
\begin{eqnarray*}
\sum_{j\in M}d\left( i,j,t_{.j}\right) &=&\sum_{j\in M}I_{i}^{d}\left(
N,\left\{ j\right\} ,t_{.j}\right) =I_{i}^{d}\left( N,M,t\right) \geq
I_{i}^{d}\left( N,M,t^{\prime }\right) \\
&=&\sum_{j\in M}I_{i}^{d}\left( N,\{j\},t_{.j}^{\prime }\right) =\sum_{j\in
M}d\left( i,j,t_{.j}^{\prime }\right) .
\end{eqnarray*}

Similarly, let $k\in N\backslash \left\{ i\right\} $. Then, 
\begin{eqnarray*}
\sum_{j\in M}d\left( k,j,t_{.j}\right) &=&\sum_{j\in M}I_{k}^{d}\left(
N,\left\{ j\right\} ,t_{.j}\right) =I_{k}^{d}\left( N,M,t\right) \leq
I_{k}^{d}\left( N,M,t^{\prime }\right) \\
&=&\sum_{j\in M}I_{k}^{d}\left( N,\{j\},t_{.j}^{\prime }\right) =\sum_{j\in
M}d\left( k,j,t_{.j}^{\prime }\right) ,
\end{eqnarray*}%
as desired.

$\left( d\right) $ 
Let $\left( N,M,t\right) ,\left( N,M,t^{\prime }\right) \in \mathcal{P}$ and 
$i\in N$ be such that $t_{i.}=t_{i.}^{\prime }$ and $t_{k.}\geq
t_{k.}^{\prime }$ for all $k\in N\backslash \left\{ i\right\} $. Let $j\in M$%
. Then, 
\begin{equation*}
I_{i}^{d}\left( N,\left\{ j\right\} ,t\right) =d(i,j,t_{.j})\leq
d(i,j,t_{.j}^{\prime })=I_{i}^{d}\left( N,\left\{ j\right\} ,t^{\prime
}\right) .
\end{equation*}

Now, given $k\in N\backslash \left\{ i\right\}$, 
\begin{equation*}
I_{k}^{d}\left( N,\left\{ j\right\} ,t\right) =d(k,j,t_{.j})\geq
d\left(k,j,t_{.j}^{\prime }\right)=I_{k}^{d}\left( N,\left\{
j\right\},t^{\prime }\right) \text{.}
\end{equation*}

By \textit{additivity}, 
\begin{equation*}
I_{i}^{d}\left( N,M,t\right) =\sum_{j\in M}I_{i}^{d}\left( N,\left\{
j\right\} ,t_{.j}\right) \leq \sum_{j\in M}I_{i}^{d}\left(
N,\{j\},t_{.j}^{\prime }\right) =I_{i}^{d}\left( N,M,t^{\prime }\right) ,
\end{equation*}%
and for each $k\in N\backslash \left\{ i\right\} ,$ 
\begin{equation*}
I_{k}^{d}\left( N,M,t\right) =\sum_{j\in M}I_{k}^{d}\left( N,\left\{
j\right\} ,t_{.j}\right) \geq \sum_{j\in M}I_{k}^{d}\left(
N,\{j\},t_{.j}^{\prime }\right) =I_{k}^{d}\left( N,M,t^{\prime }\right) .
\end{equation*}%
Thus, $I^{d}$ satisfies \textit{reciprocal artist monotonicity}.

Conversely, let $I$ be an index satisfying the three axioms. Then, as
mentioned above, $I$ is a decomposable index $I^{d}$. 
It just remains to show that $d$ satisfies the condition in the statement.
Let $i\in N\subset\mathbb{N}$, $j\in \mathbb{M}$, $x\in \mathbb{Z}_{+}^{n}$,
and $x^{\prime }_{-i}\ge x_{-i}$, 
Then, 
by \textit{reciprocal artist monotonicity}, 
\begin{equation*}
d(i,j,(x_{i},x_{-i}))=I_{i}^{d}\left( N,\{j\},(x_{i},x_{-i})\right) \geq
I_{i}^{d}\left( N,\{j\},(x_{i},x^{\prime }_{-i})\right)
=d(i,j,(x_{i},x^{\prime }_{-i})).
\end{equation*}

Similarly, let $k\in N\backslash \left\{ i\right\} $. Then, 
\begin{equation*}
d(k,j,(x_{i}^{\prime },x_{-i}))=I_{k}^{d}\left( N,\{j\},(x_{i}^{\prime
},x_{-i})\right) \leq I_{k}^{d}\left( N,\{j\},(x_{i},x_{-i})\right)
=d(k,j,(x_{i},x_{-i})),
\end{equation*}%
as desired.
\end{proof}


\section{Conclusions}

We have considered several monotonicity axioms for the problem of sharing
the revenues raised from subscriptions to music platforms among artists that
upload contents to such a platform. As mentioned earlier, monotonicity
axioms turned out to be very strong in some domains of problems. 
In our setting, 
they have far-reaching implications. 
For instance, when combined with two basic axioms (\textit{additivity} and 
\textit{null artist}), we obtain several characterizations (within the broad
family of decomposable indices). Furthermore, three of the monotonicity
axioms we consider characterize the (unappealing) egalitarian index, when
combined with two other basic axioms (\textit{symmetry} and \textit{budget
balance}). More interestingly, we have obtained characterizations for two
focal indices within the literature (the \textit{pro-rata} and \textit{%
Shapley} indices), resorting to some of those monotonicity axioms.

In contrast, another domain in which monotonicity axioms are rather weak is
that of claims problems introduced by O'Neill (1982), and surveyed by
Thomson (2003, 2015, 2019). Interestingly, a connection between claims
problems and streaming problems can be established (e.g., Berganti\~{n}os
and Moreno-Ternero, 2026a; Dietzenbacher and N\'{u}\~{n}ez-Lugilde, 2026).
Our work thus implies that the two settings (streaming problems and claims
problems) are indeed of a rather different nature.



\bigskip 

\bigskip 

\bigskip 

\bigskip 

\bigskip 

\newpage

\bigskip

\begin{thebibliography}{99}
\bibitem{} Abe, T., Lowing, D., Nakada, S., 2025. Fair allocation in
hierarchies: A compromise between marginalism and egalitarianism. Journal of
Mathematical Economics 119, 103128.

\bibitem{} Aguiar, L., 2017. Let the music play? Free streaming and its
effects on digital music consumption. Information Economics and Policy 41,
1-14.


\bibitem{ARW24} Aguiar, L., Reimers, I., Waldfogel, J., 2024. Platforms and
the transformation of the content industries. Journal of Economics \&
Management Strategy 33, 317-326.




\bibitem{AMMP22} Alaei, S., Makhdoumi, A., Malekian, A., Peke\v{c} S., 2022.
Revenue-sharing allocation strategies for two-sided media platforms:
pro-rata vs. user-centric. Management Science 68, 8699-8721.

\bibitem{BMT15} Berganti\~{n}os, G., Moreno-Ternero, J.D., 2015. The
axiomatic approach to the problem of sharing the revenue from museum passes.
Games and Economic Behavior 89, 78-92.

\bibitem{BM22} Berganti\~{n}os, G. and Moreno-Ternero J.D., 2022. Overall
monotonicity in sharing the revenues from broadcasting sports leagues.
European Journal of Operational Research 297, 338-346.

\bibitem{BM25a} Berganti\~{n}os, G. and Moreno-Ternero J.D., 2025a. Revenue
sharing at music streaming platforms. Management Science 77, 8319-8335.

\bibitem{BM25b} Berganti\~{n}os, G. and Moreno-Ternero J.D., 2025b. The
Shapley index for music streaming platforms. Information Economics and
Policy 71, 101142.

\bibitem{BM25c} Berganti\~{n}os, G. and Moreno-Ternero J.D., 2026a.
Streaming problems as (multi-issue) claims problems. European Journal of
Operational Research 329, 1030-1036.

\bibitem{BM25d} Berganti\~{n}os, G. and Moreno-Ternero J.D., 2026b. New
methods to compensate artists in music streaming platforms. Games and
Economic Behavior 159, 131-143.

\bibitem{} Chun, Y., Thomson, W., 1988. Monotonicity properties of
bargaining solutions when applied to economics. Mathematical Social Sciences
15, 11-27.


\bibitem{} Dietzenbacher, B., N\'{u}\~{n}ez-Lugilde, I., S\'{a}nchez-Rodr%
\'{\i}guez, E. 2026. Streaming platform games: construction of the core.
Annals of Operations Research forthcoming,
https://doi.org/10.1007/s10479-026-07385-6.

\bibitem{} Dietzenbacher, B., N\'{u}\~{n}ez-Lugilde, I., 2026. Solving
streaming platform problems as claims problems. Maastricht University,
Graduate School of Business and Economics. GSBE Research Memoranda No. 006 

\bibitem{} Feng, D., Ma, Z., 2026. How to fairly share streams: index
characterizations for music streaming problems. Mimeo. University of
Rochester.

\bibitem{} Flores-Szwagrzak, K., Moreno-Ternero, J.D., \O sterdal, L.P.,
2026. Technological Progress and Fair Compensation. CBS Working Paper.

\bibitem{GZ03} Ginsburgh V., Zang I., 2003. The museum pass game and its
value. Games and Economic Behavior 43, 322-325.

\bibitem{GMS24a} Gon\c{c}alves-Dosantos, J., Martinez, R., S\'{a}%
nchez-Soriano, J., 2025a. Measuring success in streaming platforms. Journal
of Economic Behavior \& Organization 232, 106941.

\bibitem{GMS24b} Gon\c{c}alves-Dosantos, J., Martinez, R., S\'{a}%
nchez-Soriano, J., 2025b. Revenue distribution in streaming. Omega 132,
103233.

\bibitem{K77} Kalai, E., 1977. Proportional solutions to bargaining
situations: interpersonal utility comparisons. Econometrica 45, 1623-1630.

\bibitem{KS77} Kalai, E., Smorodinsky, M., 1975. Other solutions to Nash's
bargaining problem. Econometrica 43, 513-518.

\bibitem{} Meggido, N., 1994. On the non-monotonicity of the bargaining set,
the kernel, and the nucleolus of a game. SIAM Journal of Applied Mathematics
27, 355-358.

\bibitem{MKAC22} Meyn, J., Kandziora, M., Albers, S., Clement, M., 2023.
Consequences of platforms' remuneration models for digital content: initial
evidence and a research agenda for streaming services. Journal of the
Academy of Marketing Science 51, 114-131.

\bibitem{} Moulin, H., Thomson, W., 1988. Can everyone benefit from growth?:
Two difficulties. Journal of Mathematical Economics 17, 339-345.

\bibitem{} O'Neill, B., 1982. A problem of rights arbitration from the
Talmud. Mathematical Social Sciences 2, 345-371.

\bibitem{} Roemer, 1986. Equality of resources implies equality of welfare.
Quarterly Journal of Economics 101, 751-784.


\bibitem{} Thomson W., 1999. Options-monotonic allocation rules. Mimeo.
University of Rochester.


\bibitem{} Thomson, W., 2003. Axiomatic and game-theoretic analysis of
bankruptcy and taxation problems: a survey, Mathematical Social Sciences 45,
249-297.


\bibitem{T15} Thomson, W., 2015. Axiomatic and game-theoretic analysis of
bankruptcy and taxation problems: An update. Mathematical Social Sciences
74, 41-59.

\bibitem{} Thomson W., 2019. \emph{How to divide when there isn't enough:
from Aristotle, the Talmud, and Maimonides to the axiomatics of resource
allocation}, Econometric Society Monograph. Cambridge University Press

\bibitem{} Thomson, W., Myerson, R., 1980. Overall monotonicity and
independence axioms, International Journal of Game Theory 9, 37-49.

\bibitem{Y85} Young, H.P., 1985. Monotonic solutions of cooperative games.
International Journal of Game Theory 14, 65-72.

\bibitem{} Young H.P., 1987. On dividing an amount according to individual
claims or liabilities, Mathematics of Operations Research 12, 398-414.

\bibitem{Y88} Young H.P., 1988. Distributive justice in taxation. Journal of
Economic Theory 44, 321-335.

\bibitem{Y26} Yu, H., 2026. On click-fraud under pro-rata revenue sharing
rule. Economic Theory Bulletin 14, 8.
\end{thebibliography}
\end{document}